\documentclass[11pt]{article}

\usepackage[margin=1in]{geometry}

\usepackage{palatino}

\usepackage{enumitem}
\setlist{itemsep=0pt,topsep=2pt}
\setitemize{itemsep=0pt,topsep=2pt}
\setenumerate{itemsep=0pt,topsep=2pt}

\usepackage[
  colorlinks,
  linkcolor = blue,
  citecolor = blue,
  urlcolor = blue]{hyperref}
\usepackage{amsmath,amssymb,amsthm}
\usepackage{mathtools}
\usepackage{color}

\newcommand{\R}{\mathbb{R}}
\newcommand{\N}{\mathbb{N}}
\newcommand{\C}{\mathbb{C}}
\newcommand{\Z}{\mathbb{Z}}

\newcommand{\ket}[1]{|#1\rangle}
\newcommand{\bra}[1]{\langle#1|}
\newcommand{\braket}[2]{\langle#1|#2\rangle}
\newcommand{\ketbras}[1]{|#1\rangle\langle#1|}

\DeclarePairedDelimiter{\set}{\lbrace}{\rbrace}
\DeclarePairedDelimiter{\abs}{\lvert}{\rvert}

\DeclarePairedDelimiter{\ceil}{\lceil}{\rceil}

\newcommand{\mc}[1]{\mathcal{#1}}

\newcommand{\id}{\ensuremath{\mathop{\rm Id}\nolimits}}
\newcommand{\ie}{\textit{i.e.}}
\newcommand{\eg}{\emph{e.g.}}
\newcommand{\ve}{\varepsilon}
\newcommand{\bs}{\mc{S}}      % Basis
\newcommand{\inA}{\mc{I}_A}   % Alice's input  
\newcommand{\inB}{\mc{I}_B}   % Bob's input  
\newcommand{\outA}{\mc{O}_A}  % Alice's output  
\newcommand{\outB}{\mc{O}_B}  % Bob's output  
\newcommand{\gh}{H}  % Hardy's pradox game
\newcommand{\ZZ}{\mathcal{U}}  
\newcommand{\ZO}{\mathcal{U}'}
\newcommand{\OZ}{\mathcal{P}}
\newcommand{\OO}{\mathcal{P}'}
\newcommand{\ZA}{A} % zero Alice
\newcommand{\ZB}{B} % zero Bob
\newcommand{\OA}{A'} % one Alice
\newcommand{\OB}{B'} % one Bob
\newcommand{\Ginf}{G^\delta_\infty}
\newcommand{\Vinf}{V^\delta_\infty}
  
\newcommand{\bips}[2]{\C^{#1}\otimes\C^{#2}}
\newcommand{\ca}{A}
\newcommand{\caa}{A'}
\newcommand{\cb}{B}
\newcommand{\cbb}{B'}
\newcommand{\cx}{X}
\newcommand{\cy}{Y}
\newcommand{\eps}{\varepsilon}

\DeclareMathOperator{\tr}{Tr}

\newcommand{\Thm}[1]{\hyperref[thm:#1]{Theorem~\ref*{thm:#1}}}
\newcommand{\Lem}[1]{\hyperref[lem:#1]{Lemma~\ref*{lem:#1}}}
\newcommand{\Cor}[1]{\hyperref[cor:#1]{Corollary~\ref*{cor:#1}}}
\newcommand{\Def}[1]{\hyperref[def:#1]{Definition~\ref*{def:#1}}}
\newcommand{\Obs}[1]{\hyperref[obs:#1]{Observation~\ref*{obs:#1}}}

\newcommand{\Sect}[1]{\hyperref[sec:#1]{Section~\ref*{sec:#1}}}

\newcommand{\Fig}[1]{\hyperref[fig:#1]{Figure~\ref*{fig:#1}}}
\newcommand{\Tab}[1]{\hyperref[tab:#1]{Table~\ref*{tab:#1}}}
\newcommand{\EqRef}[1]{\hyperref[eq:#1]{(\ref*{eq:#1})}}
\newcommand{\Eq}[1]{Equation~\hyperref[eq:#1]{(\ref*{eq:#1})}}

\newtheorem{theorem}{Theorem}

\newtheorem{lemma}[theorem]{Lemma}

\newtheorem{cor}[theorem]{Corollary}
\newtheorem*{example*}{Example}

\theoremstyle{definition}

\title{Unbounded entanglement in nonlocal games}

\author{Laura Man\v{c}inska\thanks{Centre for Quantum Technologies, National University of Singapore. Email: \texttt{laura@locc.la}} 
\and 
Thomas Vidick\thanks{Simons Institute, University of California at Berkeley. Email: \texttt{vidick@csail.mit.edu}}}

\begin{document}
\maketitle

\begin{abstract}
Quantum entanglement is known to provide a strong advantage in many two-party distributed tasks. We investigate the question of how much entanglement is needed to reach optimal performance. For the first time we show that there exists a purely classical scenario for which no finite amount of entanglement suffices. To this end we introduce a simple two-party nonlocal game $H$, inspired by Lucien Hardy's paradox. In our game each player has only two possible questions and can provide bit strings of any finite length as answer. We exhibit a sequence of strategies which use entangled states in increasing dimension $d$ and succeed with probability $1-O(d^{-c})$ for some $c\geq 0.13$. On the other hand, we show that any strategy using an entangled state of local dimension $d$ has success probability at most $1-\Omega(d^{-2})$. In addition, we show that any strategy restricted to producing answers in a set of cardinality at most $d$ has success probability at most $1-\Omega(d^{-2})$. Finally, we generalize our construction to derive similar results starting from any game $G$ with two questions per player and finite answers sets in which quantum strategies have an advantage. 
\end{abstract}

%%%%%%%%%%%%%%%%%%%%%%%
\section{Introduction}
%%%%%%%%%%%%%%%%%%%%%%%

Entanglement plays a key role in quantum information processing. The almost unnaturally strong correlations it implies were initially seen as a weird, if not undesirable~\cite{epr}, feature of quantum mechanics. Yet more recently entanglement is increasingly being thought of as a distributed resource that allows cooperating parties to accomplish otherwise impossible tasks, such as unconditionally secure cryptography~\cite{Ekert91}, randomness certification~\cite{Colbeck09,Pironio} and expansion~\cite{VV12,CoudronY13infinite}, or classical communication with improved efficiency~\cite{BennettW92superdense}. All these tasks are purely classical scenarios in which the use of shared entanglement provides a strong advantage. It is thus natural to ask how much of this new resource is needed in order to achieve optimal performance. The question arises in areas such as quantum Shannon theory~\cite{Wilde11}, communication complexity~\cite{Buhrman10}, nonlocality~\cite{Brunner13Bell}, and many others. Perhaps surprisingly, very few general results are known. Specific examples have been constructed to establish  lower bounds on the amount of entanglement required. Yet the problem of proving upper bounds is a recurrent sticking point and few such bounds are known; two-player XOR games provide a rare exception~\cite{Cleve04}. 

This unfortunate state of affairs has led to the question of whether such bounds exist in principle. The question can be abstractly formulated as follows:\\

\parbox{0.9\textwidth}{\em
Is it possible to associate an integer $d(\mathcal{S})\in\Z_+$ to any finite size $\mathcal{S}\in\Z_+$ so that any distributed task $\mathcal{T}$ of size $\mathcal{S}$ can be completed with vanishing error probability using entangled states of local dimension at most~$d(\mathcal{S})$?}
\hspace*{0.5cm}\emph{$(\star)$}\\[0.2cm]

%The key aspect of this question is whether any fixed task $\mathcal{T}$  can be associated a \emph{finite} complexity in terms of its dependence on the use of entanglement. 
We may also consider a less demanding finite-precision variant of the above question. That is, we could require to achieve $\mathcal{T}$ \emph{approximately}, within some finite precision $\eps$, and allow $d=d(\mathcal{S},\eps)$ to depend on $\eps$ as well. (For now we are purposefully leaving terms such as complexity, precision, etc., loosely defined; they will be made more concrete in the coming paragraphs.)

We study both questions in the context of nonlocal games. In such games, two separated parties are provided with questions $x$ and $y$ respectively, chosen according to a pre-specified distribution $\pi(x,y)$. Without communicating they must respectively provide answers $a$ and $b$. The task $\mathcal{T}$ is to maximize the probability that their answers satisfy a pre-determined criterion $V(a,b|x,y)=1$. The entangled value, $\omega^*(G)$, of such a game $G=(V,\pi)$ is the largest success probability achievable using finite-dimensional entangled states (see \Sect{Prelim} for precise definitions).
It is known that for any $d\in\Z^+$ there exists a nonlocal game for which any strategy attaining the optimal success probability requires measurements on a shared state with local dimension at least $d$~\cite{Bar08, Wehner08, Briet11}. Thus the upper bound $d(\mathcal{S})$ on the entanglement needed cannot be a universal constant, \ie, its value must depend on some measure of complexity of the game. 

\paragraph{Our results.} We introduce a nonlocal game $\gh$ for which the answer to question $(\star)$ is negative. While $\omega^*(H)=1$, we show that this success probability can only be achieved in the limit of strategies using entangled states of increasing dimension. More precisely, the game $\gh$ is such that for any $\ve>0$,
\begin{itemize}
\item $\gh$ can be won with probability $1-\ve$ using a shared state of local dimension $O(1/\ve^{7.3...})$;
\item any strategy that wins $H$ with probability $1-\ve$ uses a state of local dimension $\Omega(1/\sqrt{\ve})$.
\end{itemize} 
See \Lem{Strategy} and \Thm{Finite} for precise statements. The game $\gh$ is inspired by Hardy's paradox \cite{Hardy92,Hardy93}; it has two questions per party and countably infinite answer sets (see \Sect{Paradox} for a brief review of Hardy's paradox, and \Sect{Game} for a complete description of the game). We also show that in order to win $H$ with probability $1-\ve$ the quantum players must use a strategy that assigns positive probability to at least $\Omega(1/\sqrt{\eps})$ distinct answers per party (see \Thm{Answers}). Finally, starting from any two-question game $G$ with finite answer set, we construct games $\Ginf$ with the same features as the Hardy's game, as long as $0 < \delta < \omega(G)-\omega^*(G)$ (see \Sect{General}).

Our result has some bearing on the computational complexity of computing the entangled value $\omega^*(G)$ of a general nonlocal game $G$. If the input size is defined to be the total number of questions and answers in $G$, then $\omega^*(G)$ is known to be NP-hard to compute~\cite{ItoKM09}. 
The problem is not known to be in NP however. In fact, no non-trivial upper bounds on its complexity are currently known---it is even unknown whether it can be decided if $\omega^*(G)=1$. 
The only known decidability result is for a related parameter $\omega^*_{FV}(G)$, the field-theoretic value of $G$. Here the tensor-product condition on the players' strategy is relaxed to a commutativity requirement (due to this relaxation, $\omega^*(G) \le \omega_{FV}^*(G)$ for any game $G$). 
The parameter $\omega^*_{FV}$ is computable provided that the optimal commuting strategy is finite-dimensional \cite{Wehner08,Navascues08}. In such a case $\omega^*(G) = \omega_{FV}^*(G)$ since any finite-dimensional commuting strategy can be used to construct a tensor-product strategy with the same success probability.
Equality in the general case would follow from a positive resolution of Tsirelson's conjecture (see \eg~\cite{Fritz12} for a formulation and discussion of the conjecture).

Our results concerning the game $\gh$ show that no a priori upper bound on the dimension of optimal strategies exists for games with finite question and countable answer sets. Hence, in this case the ``na\"ive'' algorithm that computes $\omega^*(G)$ by performing an exhaustive search over all possible strategies in increasing dimension (and to within increasing precision) will keep finding strategies with increasing success probability. Of course, in practice one will be content with a good approximation to $\omega^*(G)$. In this case one can interpret our result as placing a lower bound, depending on the desired approximation, on the dimension in which the search must be performed. The game $H$ can thus be used as a ``dimension witness'', \emph{for any dimension}: strategies achieving success probability at least $\omega^*(H)-\eps$ must necessarily use entanglement of local dimension $d=\Omega(1/\sqrt{\eps})$. Examples of similar constructions are already known \cite{Brunner08, Slofstra11, Briet11}, but they all require using different games in order to witness increasing dimensions. In particular, for all the known constructions the number of questions increases along with the dimension to be certified.
%but they all require games with increasing number of questions in order to witness increasing dimensions. 
For instance, 
Slofstra~\cite{Slofstra11} provides an $n$-question two-answer XOR game $G_n$ for which achieving $\omega^*(G)-\eps$ requires dimension $\min(2^{\Omega(\sqrt{n})},\Omega(\eps^{-1/12}))$. Bri\"et et al.~\cite{Briet11} provide an $(1/\eps)^{\Theta(1/\eps)}$-question two-answer XOR game $G_\eps$ for which the dependence on $\eps$ is $\Omega(1/\eps)$. When compared to these constructions, our game has the drawback of having infinitely many possible answers.
%To avoid having infinitely many possible answers, we can limit our game to answers of length $\ell=\Theta(\log 1/\eps)$. 
To rectify this, we can limit our game to answers of length $\ell=\Theta(\log 1/\eps)$. 
This gives a family of games $H_\ell$ with two questions and $\text{poly}(1/\eps)$ answers. To win $H_\ell$ with probability $\omega^*(H_\ell)-\eps$, a state of dimension at least $\Omega(1/\sqrt{\eps})$ is needed (see \Cor{Witness}).

\paragraph{Related work.}
Prior to our work two examples of nonlocal games were known for which the optimal success probability can only be approached in the limit of infinite-dimensional shared entanglement. Leung et al.~\cite{Leung08} consider a game with quantum questions and answers, which has inspired  Regev and Vidick~\cite{Regev12} to produce a game with quantum questions but classical answers. The latter game has two possible $(3\times 3)$-dimensional states as questions, and two possible answers per player. Our game provides the first example with classical questions and answers, thus resolving a question first formally asked in~\cite{Navascues08}. The same question was also asked~\cite{Wehner08}, but specifically for games with finite question and answer sets. Although there is some numerical evidence that infinite-dimensional entanglement may be needed in that case as well~\cite{Pal10}, the question remains tantalizingly open.

\paragraph{Organization of the paper.}
The remainder of the paper is organized as follows. In \Sect{Prelim} we formally define the classical and entangled values of a nonlocal game. In \Sect{Hardy} we review Hardy's paradox, introduce the game $H$, and show that the classical value of $H$ is $3/4$ while its entangled value equals $1$. We also give a more general construction which yields a game with similar properties to Hardy's, but starting from any game with two questions per player and finite answer sets. In \Sect{Finite} we establish our main result by showing that restricting the dimension of shared entanglement bounds the achievable success probabilities away from one (see~\Thm{Finite}). We also show that restricting the number of different answers has the same effect (see~\Thm{Answers}). We conclude in \Sect{Conclusion}.

\section{Preliminaries}
\label{sec:Prelim}
In this section we explain the terminology used to discuss one-round two-player nonlocal games. A reader familiar with nonlocal games is encouraged to proceed directly to the next section.

A two-party nonlocal game $G$ consists of a probability distribution $\pi$ over a set of the form $\inA \times \inB$, called the set of questions; sets of answers $\mc{O}_A$ and $\mc{O}_B$, and a verification function $V: \mc{O}_A \times \mc{O}_B \times \mc{I}_A \times \mc{I}_B  \to \set{0,1}$. 
%We only consider the case where the sets $\inA,\inB$ are countable. 
With probability $\pi(x,y)$ the referee sends the two players, traditionally called Alice and Bob, questions $x\in\inA$ and $y\in\inB$, respectively. Without communicating, the players must produce answers $a\in\outA$ and $b\in\outB$, respectively. They win if $V(a,b,x,y)=:V(a,b|x,y)=1$ and lose otherwise. Classical players can use shared randomness to possibly enhance their strategy while quantum players can use shared entanglement. The goal of the players is to maximize their probability of winning. In case of classical players, we call this probability the \emph{(classical) value} of game G, denoted $\omega(G)$. In the quantum case we call it the \emph{entangled value} of $G$, denoted $\omega^*(G)$. Since quantum players are at least as powerful as classical ones, $\omega(G)\leq\omega^*(G)$ for any game $G$. 

\subsection{\large Classical strategies and value}
 
Any classical strategy for a game $G$ can be specified using a pair of  functions, 
\begin{equation}
  \alpha : \inA \times R \to \outA \;\;\text{ and }\;\;
  \beta  : \inB \times R \to \outB,
\end{equation}
where $R$ consists of the possible values of shared randomness (which, without loss of generality, also includes any private randomness). We define $\alpha(x,r)$ to be Alice's answer on question $x$ given that the shared randomness takes value $r$. We define $\beta(y,r)$ similarly. If the shared randomness is distributed according to $\tau: R \to [0,1]$, then the classical value of the game is
\begin{equation}
  \omega(G) := \sup_{\alpha,\beta} \omega(G|\alpha,\beta) 
  :=  \sup_{\alpha,\beta} \sum_{r} \tau(r) 
  \sum_{x,y} 
  \pi(x,y) V\Big(\alpha_r(x),\beta_r(y)|x,y\Big),
\end{equation}
where  $\alpha_r(x):=\alpha(x,r)$ and $\beta_r(y):=\beta(y,r)$. Note that for any $r\in R$ the pair $(\alpha_r,\beta_r)$ specifies a deterministic strategy. Since $\omega(G|\alpha,\beta)$ is a convex combination of $\omega(G|\alpha_r,\beta_r)$, there exists some $r\in R$ for which $\omega(G|\alpha_r,\beta_r) \geq \omega(G|\alpha,\beta)$. Therefore, it suffices to consider only deterministic strategies $(\alpha,\beta)$, \ie,
\begin{equation}
  \omega(G) = \sup_{\substack{\alpha: \inA \to \outA \\
                             \beta: \inB \to \outB}} \;
  \sum_{x,y} 
  \pi(x,y) V\big(\alpha(x),\beta(y)|x,y\big).
\label{eq:CStrat}
\end{equation}
The following two lemmas show that if a game either has finitely many questions, or finitely many answers, then the classical value can always be achieved exactly, \ie, the supremum in~\eqref{eq:CStrat} can be replaced with a maximum.

\begin{lemma}
Let $G=(V,\pi)$ be a two-party game with finite question sets $\inA$, $\inB$ and (arbitrary) answer sets $\outA,\outB$. Then the classical value of $G$ is always achieved by some deterministic strategy, \ie, $\omega(G) = \omega(G|\alpha,\beta)$ for some $\alpha: \inA \to \outA$ and $\beta: \inB \to \outB$.
\end{lemma}

\begin{proof}
Consider a two-party game $G=(V,\pi)$ with finite question sets.
We have already argued that it suffices to consider only deterministic strategies to find the classical value $\omega(G)$. Therefore, it only remains to show that some deterministic strategy $(\alpha,\beta)$ achieves success probability $\omega(G)$, \ie, $\omega(G) = \omega(G | \alpha, \beta)$.
Since $V(a,b|x,y)\in\set{0,1}$ and the sets $\inA,\inB$ are finite, the sum in \Eq{CStrat} can only take on a finite number of different values, each of which is  attained by a particular pair of strategies. Therefore, the supremum can be replaced with the maximum and the value of $G$ is achieved by some deterministic strategy. 
\end{proof}

\begin{lemma}
Let $G=(V,\pi)$ be a 2-party game with countably infinite question sets $\inA$, $\inB$ and finite answer sets $\outA,\outB$. Then the classical value of $G$ is always achieved by some deterministic strategy.
\label{lem:InfiniteInput}
\end{lemma}
\begin{proof}
Consider a game $G=(V,\pi)$ with question sets $\inA,\inB=\set{1,2,\dotsc}$ and finite answer sets $\outA,\outB$. We can identify any deterministic strategy $(\alpha,\beta)$ with an infinite word 
\begin{equation}
  \big(\alpha(1),\beta(1)\big)\big(\alpha(2),\beta(2)\big)
  \big(\alpha(3),\beta(3)\big)\cdots
\end{equation}
over the finite alphabet $\outA \times \outB$. Consider a sequence of deterministic strategies $(\alpha_n,\beta_n)$ with success probabilities converging to $\omega(G)$. Let $w_n$ be the $n$-letter prefix of the word corresponding to strategy $(\alpha_n,\beta_n)$ and consider the set of words $L=\set{w_n}_n$. K\"onig's lemma states that if $X$ is an infinite prefix-closed set of words over a fixed alphabet, then we can find an infinite word $x$ whose any prefix is shared by some word in $X$ (see \eg{} Proposition~1.2.3 in \cite{Lothaire}). 
%Since $L$ is an infinite set of words over a finite alphabet, we can construct an infinite word $w$ whose any prefix is shared by some word in $L$ (this is essentially K\"onig's lemma). 
Since the set $L$ is prefix-closed by construction, there exists $w$ whose any prefix is shared by some word in $L$.
It now remains to note that the strategy that corresponds to $w$ achieves success probability $\omega(G)$. This is because the partial strategies corresponding to prefixes of $w$ achieve success probabilities arbitrarily close to $\omega(G)$.
\end{proof}

In the case when both the question sets and answer sets are allowed to be infinite the value of the game cannot always be achieved. To see this consider the following game, which was suggested to us by A.~Belovs. The question and answer sets are $\inA,\inB,\outA,\outB :=\N$, $\pi(x,y) = 2^{-x-y}$ and $V(a,b|x,y) = 1$ if and only if $a,b>\max(x,y)$. Clearly no deterministic strategy achieves success probability one. Yet if we define $(\alpha_n,\beta_n)$ via $\alpha(x)=x+n$ and $\beta_n(y)=y+n$, then  $(\alpha_n,\beta_n)$ gives a sequence of strategies with vanishing error. Hence, $\omega(G)=1$ even though no single strategy can be used to win game $G$ perfectly.

\subsection{\large Quantum strategies and value}
A quantum strategy is specified by a finite-dimensional shared state $\ket{\psi}\in\C^{d_A}\otimes\C^{d_B}$, a POVM $\mc{A}^{x} = \set{A_a^{x} : a\in\outA}$ for every question $x\in\inA$ to Alice, and a POVM $\mc{B}^{y} = \set{B_b^{y} : b\in\outB}$ for every question $y\in\inB$ to Bob. Upon receiving questions $x$ and $y$, the parties measure their systems with POVMs $\mc{A}^{x}$, $\mc{B}^{y}$ respectively and answer with their respective measurement outcomes $a$ and~$b$. The dimension of a strategy is the dimension of the shared entangled state it employs. We also use the term ``local dimension'' to refer to $d_A$ or $d_B$, \ie, the dimension of Alice's or Bob's local systems. 

The quantum value of the game is given by
\begin{equation}
  \omega^*(G) = \sup
  \sum_{x,y} \pi(x,y) V(a,b|x,y) \bra{\psi} 
  \big(A_a^x \otimes B_b^y\big) \ket{\psi},
\end{equation}
where the supremum is taken over all finite-dimensional shared states $\ket{\psi}$ and sets of POVMs $\set{\mc{A}^x}_x$ and $\set{\mc{B}^y}_y$.

\section{Hardy's game and its generalization}
\label{sec:Hardy}

We introduce a nonlocal game that we call Hardy's game, as it is based on Hardy's paradox~\cite{Hardy92,Hardy93}. We first describe the paradox, then the game we derive from it, and exhibit a sequence of good strategies for quantum players of this game. We end this section by giving a general construction which yields a nonlocal game with similar properties as Hardy's game, starting from any game $G$ with finite question and answer sets. 

\subsection{Hardy's paradox}
\label{sec:Paradox}

Hardy's paradox is a two-party experimental setup whose outcomes, as predicted by quantum mechanics, cannot be reproduced by any local hidden variables theory. The setup is noteworthy for being minimal in a sense that it involves only two qubits, on which both of the parties perform two measurements with two outcomes each. Our description of the setup roughly follows the explanations given in \cite{Laloe01}.

Suppose Alice and Bob share a two-qubit state
\begin{equation}
 \ket{\psi} :=  \frac{\sin(\theta) \ket{v_0',v_0'}
   -\cos(\theta) \big(\ket{v_0',v_1'}+ \ket{v_1',v_0'} \big)}{\sqrt{1+\cos^2(\theta)}},            
\label{eq:Psi}
\end{equation}
where $\bs'= \set{\ket{v_0'},\ket{v_1'}}\subseteq\R^2$ is any orthonormal basis and $0<\theta<\frac{\pi}{2}$.
Let $\bs = \set{\ket{v_0},\ket{v_1}}\subseteq\R^2$ be the basis obtained by rotating $\bs'$ by angle $2\theta$ in the $xz$ plane of the Bloch sphere, \ie,
\begin{align}
  &\ket{v_0} = \cos(\theta)\ket{v_0'} + \sin(\theta) \ket{v_1'},\\
  &\ket{v_1} = -\sin(\theta)\ket{v_0'} + \cos(\theta) \ket{v_1'}.
\label{eq:Basis}
\end{align}
If Alice and Bob measure the state $\ket{\psi}$ using the bases $\bs$ or $\bs'$ the following conditions are satisfied (see also \Fig{Diagram}):
\begin{enumerate}[label=(\roman*)]
\item if both parties measure in $\bs$, the outcome pair $(0,0)$ occurs with probability $p_\theta>0$ (see \Eq{P} for an exact formula);
\item if one of the parties measures in $\bs$ and the other in $\bs'$, the outcome pair $(0,0)$ never occurs;
\item if both parties measure in $\bs'$, the outcome pair $(1,1)$ never occurs.
\end{enumerate}
\begin{figure}[h!]
\centering
\includegraphics{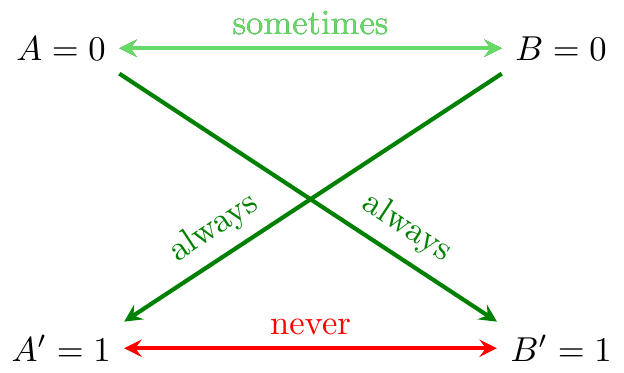}
\caption{\small Summary of correlations between the outcomes of Alice and Bob's measurements in Hardy's setup. For instance, the arrow $\ca=0\to\cbb=1$ marked ``always'' means that, if Alice measures her system in basis $\bs$ and Bob measures his system in basis $\bs'$, then whenever Alice obtains the outcome $0$ Bob always obtains the outcome $1$.}
\label{fig:Diagram}
\end{figure}

To see that condition~(i) is satisfied, we note that
\begin{equation}
  p_{\theta} := \abs{\braket{v_0,v_0}{\psi}}^2 = 
  \frac{\cos^4(\theta) \sin^2(\theta)}{1+\cos^2(\theta)} > 0
\label{eq:P}
\end{equation}
as $0<\theta<\pi$. To verify the other two conditions one can check that $\braket{v_0^{\phantom{'}},v_0'}{\psi}=0$, $\braket{v_0',v_0^{\phantom{'}}}{\psi} = 0$, and $\braket{v_1',v_1'}{\psi}=0$.

Any two-qubit state that is neither product nor maximally entangled can be expressed in the form of \Eq{Psi} for appropriately chosen basis $\bs'$ and angle $\theta$ \cite{Hardy93,Goldstein94}. Therefore, almost any two-qubit state can be used to perform an experiment whose outcomes will satisfy conditions (i)--(iii). Different values of $\theta$ will result in different $p_\theta$; one can verify that the maximum $p_\theta$ is $\tfrac{5\sqrt{5}-11}{2}\approx0.09$ and is attained at $\theta=\arccos\big(\big(\tfrac{\sqrt{5}-1}{2}\big)^{1/2}\big) < \frac{\pi}{4}$.

Local hidden variables theories cannot reproduce the predictions (i)--(iii) made by quantum mechanics. Any such theory assigns definite outcomes (depending only on the hidden variables) for each of the four measurements: $\ca,\caa,\cb$ and $\cbb$. Condition~(i) guarantees that for some setting of the hidden variables the outcomes associated with $A$ and $B$ will both be $0$. According to Condition~(ii) this implies that for the same setting of hidden variables the outcomes associated with $\cbb$ and $\caa$ must both be $1$. This, however, contradicts Condition~(iii).

A different construction that is often used to disprove the existence of non-contextual hidden variables theories are the so-called Kochen-Specker sets \cite{Gleason57, Kochen67}. These sets can easily be turned into a nonlocal game with entangled value $1$ and classical value bounded away from $1$~\cite{Cleve04}. Hardy's paradox in itself does not immediately yield such a game; indeed it is unclear how to enforce Condition (i) in the framework of nonlocal games. In the next section we show how to accommodate with that condition and turn Hardy's paradox into a nonlocal game.

\subsection{A nonlocal game from Hardy's paradox}
\label{sec:Game}
Building on Hardy's paradox, we now propose the following nonlocal game that we call \emph{Hardy's game} and denote as~$\gh$.
In game $\gh$ there are two questions per player: Alice's question set is $\mc{I}_A = \set{\ca,\ca'}$, Bob's question set is $\mc{I}_B = \set{\cb,\cb'}$, and the questions are uniformly distributed, \ie, $\pi(\cx,\cy) = \frac{1}{4}$ for all $\cx\in\inA$, $\cy\in\inB$. The possible answers for both parties are binary strings of arbitrary but finite length, \ie, $\mc{O}_A,\mc{O}_B = \set{0,1}^*$. The verification function $V: \mc{O}_A \times \mc{O}_B \times \mc{I}_A \times \mc{I}_B  \to \set{0,1}$ is defined by $V(a,b|\cx,\cy)=1$ if and only if all of the following conditions are satisfied.
\begin{enumerate}
\item The answer strings $a$ and $b$ have the same length, \ie, $\abs{a} = \abs{b}$.
\item If $\cx=\ca$ and $\cy=\cb$, then $a_i = b_i =0$ for some position $i\in [n]$, where $n:=\abs{a}=\abs{b}$.
\item For each position $i\in[n]$:
\begin{enumerate}
\item if $(\cx,\cy)=(\ca,\cb')$ or $(\cx,\cy) = (\ca',\cb)$, then $a_i=1$ or $b_i=1$;
\item if $(\cx,\cy) = (\ca',\cb')$, then $a_i=0$ or $b_i=0$.
\end{enumerate}
\end{enumerate}
In the above, Condition~3 requires that each pair of answer bits $(a_i,b_i)$ satisfies the last two conditions in Hardy's paradox. Condition~2 requires that there exists a position $i$ for which $(a_i,b_i)$ satisfy the first condition in Hardy's paradox with certainty.

\subsection{The classical and entangled values}
\label{sec:Values}

In this section we determine both the classical value and the entangled value of Hardy's game. 

\begin{lemma}
The classical value of Hardy's game is $\omega(\gh) = 3/4$.
\label{lem:HCvalue}
\end{lemma}

\begin{proof}
The value of Hardy's game is given by
\begin{equation}
  \omega(\gh) = \max_{\alpha,\beta} 
  \sum_{X\in\set{\ca,\caa}}\sum_{Y\in\set{\cb,\cbb}}
  \frac{1}{4} V\big(\alpha(X),\beta(Y)| X,Y\big)
\end{equation}
and thus $\omega(\gh)\in\set{0,\frac{1}{4},\frac{1}{2},\frac{3}{4},1}$.
In fact, $\omega(\gh)\ge \frac{3}{4}$, since a strategy defined via $\alpha(\ca) = \beta(\cb) = 0$ and $\alpha(\caa) = \beta(\cbb) = 1$ achieves a success probability of $\frac{3}{4}$. To complete the proof it remains to rule out the existence of a perfect deterministic strategy. For contradiction, assume that $\omega(\gh|\alpha,\beta) = 1$ for some deterministic strategy $(\alpha,\beta)$. Since the strategy is perfect, there exists $i$ such that the $i$th bit of the strings $\alpha(\ca)$ and $\beta(\cb)$ is zero. This implies that the $i$th bit of both strings $\alpha(\caa)$ and $\beta(\cbb)$ equals one. Therefore, the players lose on question $(\caa,\cbb)$ which gives the desired contradiction.
\end{proof}

\begin{lemma}
The entangled value of Hardy's game is $\omega^*(\gh) = 1$. Furthermore, for any $\eps>0$ there exists a quantum strategy that succeeds with probability at least $1-\eps$ using only answers of length $\ell=\Theta(\log(1/\eps))$ and entangled state of local dimension $d = \Theta(\eps^{1/\log c}) = \Theta(1/\eps^{7.35\ldots})$, where $c=\tfrac{13-5\sqrt{5}}{2}$.
\label{lem:Strategy}
\end{lemma}

\begin{proof}
For each $n\in\N$ consider a strategy $\mathfrak{S}_n$ in which players share $n$ copies of the state $\ket{\psi}$ defined in~\Eq{Psi}, for some value of $\theta$. To produce an $n$-bit answer string they measure each of the $n$ copies in basis $\bs = \set{\ket{v_0},\ket{v_1}}$ upon receiving an unprimed question and in basis $\bs'$ (see \Eq{Basis}) upon receiving a primed question. To analyze the success probability of strategy $\mathfrak{S}_n$ recall the three conditions from \Sect{Game}. Since both players answer strings of length $n$, Condition~1 is always satisfied. As discussed in \Sect{Paradox}, the outcomes obtained by measuring $\ket{\psi}$ in basis $\bs$ and $\bs'$ satisfy the conditions from Hardy's paradox. This implies that Alice's and Bob's answer bits $(a_i,b_i)$ always satisfy Condition~3. Finally, note that upon question $(\ca,\cb)$ we have $a_i=b_i=0$ with probability $p_\theta$ (see \Eq{P}). Hence, Condition~2 is satisfied with probability $1-(1-p_\theta)^n$. If $\theta=\arccos\big(\big(\tfrac{\sqrt{5}-1}{2}\big)^{1/2}\big)$, then $\mathfrak{S}_n$ errs with probability $\tfrac{1}{4}\big(\tfrac{13-5\sqrt{5}}{2}\big)^n$. Let $c=\tfrac{13-5\sqrt{5}}{2}\approx0.91<1$, fix any $\eps>0$ and consider $n=\ceil[\big]{\tfrac{\log \eps}{\log c}} = \Theta(\log(1/\eps))$. Then $\mathfrak{S}_n$ errs with probability at most $\eps$, answers strings of length $\Theta(\log(1/\eps))$ and uses entanglement of local dimension $d=2^n = \Theta(\eps^{1/\log c})$.  
\end{proof}

\subsection{A general construction}\label{sec:General}

The construction we used to define Hardy's game from Hardy's paradox can also be applied when starting from any two-player game $G$ with finite question and answer sets. For any such game $G=(V,\pi)$ and any $\delta>0$ we define a game $\Ginf=(V_\infty^\delta,\pi)$. Results similar to those in \Sect{Values} can be obtained for any of the defined games. Yet we focus on the games $\Ginf$ arising from games $G$ with two questions per player, denoted $A,A'$ and $B,B'$ respectively, and $\delta<\omega^*(G)-\omega(G)$, since our main contribution in \Sect{Finite} only applies in this case. For example, our construction can be applied to the CHSH game, to obtain CHSH$_\infty^\delta$ for any $0<\delta<\cos^2(\frac{\pi}{8})-\frac{3}{4}$.

Fix some entangled strategy that succeeds at $G$ with probability $\omega^*(G)-\frac{\delta}{2}$ and let $\omega^*_{XY}$ be its success probability conditioned on questions $(X,Y)$ being asked. The construction of the game $\Ginf$ is very simple. In the game, the verifier first chooses questions $(X,Y)\in\set{\ZA,\OA}\times\set{\ZB,\OB}$ according to $\pi$, and sends $X$ to Alice and $Y$ to Bob. Each player may reply with a string of arbitrary yet finite length over the respective answer alphabets in $G$, \ie, $\outA$ and $\outB$. The players win, \ie, $V_\infty^\delta(a,b|X,Y)=1$, if and only if their respective answer strings $a=(a_i)\in\outA^*$ and $b=(b_i)\in\outB^*$
have the same length and the fraction of positions $i$ for which $V(a_i,b_i|X,Y)=1$ is at least $\omega^*_{XY}-\frac{\delta}{2}$.
The above definition of $\Ginf$ depends on the probabilities $\omega^*_{XY}$ and hence one should rather write $\Ginf(\omega^*_{00},\omega^*_{10},\omega^*_{01},\omega^*_{11})$. For the sake of simplicity from now on we concentrate on the case where it is possible to choose a strategy for which the success probability does not depend on the question $(X,Y)$, \ie, $\omega^*_{XY} = \omega^*(G)-\frac{\delta}{2}$. 

Our results for Hardy's game in \Sect{Values} have straightforward extensions to the general case. More precisely, for any game $G$ with finite question and answer sets and any $\delta>0$, 
we have $\omega(\Ginf)=\omega(G)$,  $\omega^*(\Ginf)=1$ and the players can achieve success probability $1-\eps$ using entanglement of local dimension $\text{poly}\big(\frac{1}{\eps}\big)$.

\section{Finite strategies do not achieve the entangled value}
\label{sec:Finite}

The results in this section apply to Hardy's game as well as any game $\Ginf$ which is constructed from some game $G$ which has \emph{two questions} per player and finite answer sets $\outA$ and $\outB$ as explained in \Sect{General}. (We suspect that this restriction can be relaxed but did not investigate this further.) We show that for all small enough $\delta>0$ the entangled value cannot be attained by any finite-dimensional strategy. We also show that the same holds for strategies with bounded answer length (irrespective of the dimension of the entangled state they are based on).

Intuitively, the need for unbounded amount of entanglement to successfully play $\Ginf$ stems from the requirement that the answer strings in $\Ginf$ should on average 
%achieve the entangled value of $G$.
achieve the value of at least $\omega^*(G)-\frac{\delta}{2}$.
In order to guarantee success the players are ``required'' to repeatedly sample from an entangled strategy for $G$, necessitating more and more entanglement as they improve their chances by exploiting concentration. 

At a slightly more formal level, our proof uses arguments that are reminiscent of techniques used to lower bound the amount of quantum communication required to send classical information from Alice to Bob in quantum one-way communication setting. We show that increasingly better quantum strategies for $\Ginf$ give rise to residual states (post-measurement states on Bob's side after Alice's measurement) that can be distinguished with increasing success probability. We obtain the desired bound on the dimension by observing that the latter can only happen if the states are of sufficiently high dimension.
The precise argument is slightly more involved as indeed success in the game does not imply any form of communication, and we cannot use the intuition directly. Details are given in the proof of Theorem~\ref{thm:Finite} in the next section.

\subsection{Strategies with bounded entanglement}

\begin{theorem}
Let $\eps>0$. Consider a strategy which succeeds at Hardy's game $H$ with probability at least $1-\eps$ and uses an entangled state of local dimensions $d_A$ and $d_B$ respectively. Then $\min(d_A,d_B)\geq 1/(16\sqrt{\eps})$. As a consequence, the game $H$ cannot be won with probability one using finite-dimensional entanglement.

Furthermore, a similar statement applies for any game $\Ginf$ with $0<\delta<\omega^*(G)-\omega(G)$, constructed from a two-player game $G$ with two questions per player and finite answer sets. In this case, $\min(d_A,d_B)=O\big(\frac{1}{\sqrt{\eps}}\big)$, where the implied constant depends on $\min_{XY}\pi(X,Y)$ only.\footnote{For games with two questions per player $\min_{XY}\pi(X,Y)$ is necessarily positive, unless $\omega^*(G)=\omega(G)$. To see this, suppose $G$ is such that \eg~$\pi(A',B')=0$. Fix an optimal quantum strategy. Consider a classical strategy in which each player, upon receiving the unprimed question, uses shared randomness to sample from the distribution produced by the quantum strategy upon questions $(A,B)$. Now, upon receiving the primed question a player knows that the other player received the unprimed question. Hence the player can use the same shared randomness to again simulate the behavior of the quantum players: for any value of the shared random string the player knows precisely which answer the other player will provide. Hence the classical strategy reproduces the quantum strategy perfectly.} 
\label{thm:Finite}
\end{theorem}

Below we give a proof of the ``furthermore'' part of \Thm{Finite} which applies for the games $\Ginf$. Although, strictly speaking, Hardy's game does not fall in this category, the proof is in all points similar. Before proceeding we give a key lemma with the following informal interpretation. Consider any pair of answers $(s,t)$ associated with Alice's questions $A$ and $A'$ respectively in $\Ginf$. 
Then we can find a question $Y\in\set{B,B'}$ for Bob such that for any $u\in\outB^*$ either the answer pair $(s,u)$ is rejected upon question $(A,Y)$ or the answer pair $(t,u)$ is rejected upon question $(A',Y)$.
To state the lemma, let us define the following sets. For any $s,t\in\outA^*$ let
\begin{equation}\label{eq:def-bs}
  \ZZ_s := \{u\in\outB^*:\,  
  \text{ $\Vinf(s,u|\ZA,\ZB)=0$, \ie, $(s,u)$ is rejected upon questions $(\ZA,\ZB)$}\}
\end{equation}
and let $\OZ_t$ be the set of answers $u$ for Bob such that $(t,u)$ is rejected upon question $(\OA,\ZB)$. Similarly define $\ZO_s$ and $\OO_t$ as the sets of Bob's answers corresponding to rejected answer pairs for questions $(\ZA,\OB)$ and $(\OA,\OB)$ respectively.\footnote{The letters $\mathcal{P}$ and $\mathcal{U}$ are meant to stand for the primed and unprimed questions for Alice.}

\begin{lemma}\label{lem:discriminate}
For any $(s,t)\in (\outA^* \times \outA^*)$ it holds that either $\ZO_s \cup \OO_t = \outB^*$ or $\ZZ_s \cup \OZ_t = \outB^*$.  
\end{lemma}

\begin{proof}
For contradiction suppose that there exists a pair of answer strings $(s,t)\in(\outA^* \times \outA^*)$ such that both $\ZO_s \cup \OO_t \subsetneq \outB^*$ and $\ZZ_s \cup \OZ_t \subsetneq \outB^*$. Hence there exists a pair of strings of answers $(u,v)\in(\outB^* \times \outB^*)$ for Bob such that $u\notin \ZZ_s \cup \OZ_t$ and $v\notin \ZO_s \cup \OO_t$. By definition, this means that the pair of answers $(s,u)$, $(s,v)$, $(t,u)$ and $(t,v)$ are accepted by the verifier in the game $\Ginf$, on questions $(\ZA,\ZB)$, $(\ZA,\OB)$, $(\OA,\ZB)$ and $(\OA,\OB)$ respectively. We now derive a contradiction by defining a classical strategy for the game $G$ that has success probability strictly larger than $\omega(G)$. The strategy is simple: upon receiving their respective questions, the players use shared randomness to select an integer $i\in \{1,\ldots,|s|\}$ (note that necessarily $|s|=|t|=|u|=|v|$) and answer the $i$-th symbol in their strings $s$ or $t$ for Alice (depending on whether her question was $\ZA$ or $\OA$), and $u$ or $v$ for Bob (depending on whether his question was $\ZB$ or $\OB$). By definition of $V_\infty^\delta$, their success probability is at least $\omega^*(G)-\delta>\omega(G)$, a contradiction.
\end{proof}

Based on \Lem{discriminate} one can already give a short proof of the fact that no finite strategy can achieve success probability $\omega^*(\Ginf)=1$ for any of the considered games $\Ginf$. 
To see this, assume that a perfect finite strategy exists and let $\rho_s$ be Bob's residual state in the case when Alice has responded with string $s$ upon question $A$ and let $\rho'_t$ be his state in the case when Alice has responded $t$ upon question $A'$. Note that if $\ZZ_s \cup \OZ_t = \outB^*$ then any answer of Bob's which is accepted upon question $(A,B)$ and Alice answering $s$ would be rejected upon question $(A',B)$ and Alice answering $t$. Hence, in order to win Bob must be able to perfectly distinguish state $\rho_s$ from $\rho_t$, \ie, $\rho_s \perp \rho_t$. \Lem{discriminate} allows to conclude that $\rho_s\perp \rho_t$ for all $s,t$ which is a contradiction since $\sum_s \rho_s = \sum_t \rho_t$ cannot be the zero matrix.

To prove the dimension estimate stated in \Thm{Finite} we need a more involved quantitative argument which we are about to give. Nevertheless, the underlying idea remains the same.

\begin{proof}[Proof of \Thm{Finite}]
Fix $\eps>0$, and assume that there exists a quantum strategy for winning $\Ginf$ with probability $1-\eps$ using $\ket{\psi}\in\bips{d_A}{d_B}$ as the shared entanglement. Without loss of generality we can assume that $\ket{\psi}$ has full Schmidt rank and  $d_A=d_B=:d$. Next, let $\set[\big]{A_s : s\in\outA^*}$ be the POVM measured by Alice upon question $\ca$. Define POVM $\set[\big]{A'_t : t\in\outA^*}$ similarly, as well as POVMs $\set[\big]{B_u : u\in\outB^*}$ for Bob on question $B$ and  $\set[\big]{B'_v : v\in\outB^*}$ on question $B'$. 
Also, for any $s,t\in\outA^*$, define the following post-measurement states on Bob's system
\begin{equation}
  \rho_s := \tr_A\big((A_s \otimes \id)\ketbras{\psi}\big) 
  \;\;\text{and}\;\;
  \rho'_{t} := \tr_A\big((A'_t \otimes \id)\ketbras{\psi}\big),
\end{equation}
and let $\rho := \sum_s \rho^{\phantom{i}}_s = \sum_t \rho'_t=\tr_A(\ket{\psi}\bra{\psi})$.
Recall the definition of the set $\mathcal{U}'_s$ in \Eq{def-bs}, and define $U'_s := \sum_{v\in\mathcal{U}'_s} B_v$. Similarly define $U_s := \sum_{u\in\mathcal{U}_s} B_u$, $P'_t:= \sum_{v\in\mathcal{P}'_t} B_v$ and $P_t := \sum_{u\in\mathcal{P}_t} B_u$. 
Our assumption that the strategy succeeds with probability $1-\eps$ implies that the following four relations must hold:
\begin{align}
 \pi(A,B) \sum_{s}\, \tr(\rho_s \,U_s ) &\leq \eps\label{eq:succ-1}\\
 \pi(A,B') \sum_{s} \,\tr(\rho_s \,U'_s ) &\leq \eps\label{eq:succ-2} \\
 \pi(A',B) \sum_{t} \,\tr(\rho'_t \,P_t ) &\leq \eps\label{eq:succ-3}\\
 \pi(A',B') \sum_{t} \,\tr(\rho'_t \,P'_t ) &\leq \eps\label{eq:succ-4}.
\end{align}
To proceed we define the following sets:
\begin{align*}
  S' &:= \set[\big]{(s,t)\in (\outA^*\times \outA^*) :\, 
  \mathcal{U}'_s \cup \mathcal{P}'_t = \outB^*}
  \\
  S &:= \set[\big]{(s,t)\in(\outA^*\times \outA^*) \setminus S' :\, 
  \mathcal{U}_s \cup \mathcal{P}_t = \outB^*}.
\end{align*}
By the definition of the two sets and \Lem{discriminate}, the set $\outA^*\times\outA^*$ is the disjoint union of $S$ and $S'$.
We derive a lower bound on $d$ by studying the following quantity:
\begin{equation}\label{eq:def-rho}
\tr(\rho^2)\,=\,\sum_{s,t\in\outA^*} \tr(\rho^{\phantom{i}}_s\rho'_t)\,=\,  
  \sum_{(s,t)\in S'} \tr(\rho^{\phantom{i}}_s\rho'_t)+
  \sum_{(s,t)\in S} \tr(\rho^{\phantom{i}}_s\rho'_t).
\end{equation}
Since $\rho$ is a $d$-dimensional density matrix it holds that $\tr(\rho^2)\geq 1/d$. We now upper bound each of the two terms on the right-hand side of \Eq{def-rho} as a function of $\eps$. The two bounds are similar; we analyze the summation corresponding to pairs $(s,t)\in S'$. 

By the definition of the set $S'$, we have that $\mathcal{U}'_s \cup \mathcal{P}'_t = \outB^*$ for any $(s,t)\in S'$. Hence we can decompose
\begin{equation}\label{eq:double-count}
\id \,=\, U'_s+P'_t-R'_{s,t},
\end{equation}
where the term $R'_{s,t}:= \sum_{v\in\mathcal{U}'_s\cap \mathcal{P}'_t} B'_v$ takes care of double counting.
Write
\begin{align}
\sum_{(s,t)\in S'} \tr(\rho_s\rho'_t) &= \sum_{(s,t)\in S'}\tr(\rho_s(U'_s+P'_t-R'_{s,t})\rho'_t)\notag\\
&=  \sum_{(s,t)\in S'}\tr(\rho_s U'_s\rho'_t) +  \sum_{(s,t)\in S'}\tr(\rho_s P'_t\rho'_t) -\sum_{(s,t)\in S'}\tr(\rho_s R'_{s,t}\rho'_t).\label{eq:s0-1}
\end{align}
We bound each of the three terms from \Eq{s0-1} separately. The first two can be bounded in a similar way. Specifically, for the first term, if we define $\sigma_s := \sum_{t:\,(s,t)\in S'} \rho'_t$ then
\begin{align*}
  \sum_{(s,t)\in S'} \tr(\rho_s U'_s\rho'_t)
  &\leq \sum_{s} \abs[\big]{\tr(\rho_s U'_s\sigma_s)}
\\
  &\leq \sum_{s} \Big( \tr(\rho_s U'_s\sigma_s^2 U'_s) \Big)^{1/2}
  \Big(\tr(\rho_s)\Big)^{1/2}
\\
  &\leq \Big(\sum_s \tr(\rho_s U'_s\sigma_s^2 U'_s)\Big)^{1/2}
  \Big(\sum_s \tr(\rho_s)\Big)^{1/2} 
\\
  &\leq  \Big(\sum_s \tr(\rho_s U'_s)\Big)^{1/2}
\\
  &\leq \sqrt{\eps/\pi(A,B')},
\end{align*}
where the second inequality uses the Cauchy-Schwarz inequality for the trace inner product (using cyclicity of the trace to write $\tr(\rho_s U'_s\sigma_s) = \tr(\sqrt{\rho_s} U'_s\sigma_s\cdot \sqrt{\rho_s})$), the fourth uses $\sigma_s^2 \leq \id$, $(U'_s)^2\leq U'_s$, and $\tr(\rho)\leq 1$, and the last is by~\eqref{eq:succ-2}. The same bound can be derived for the second term in \Eq{s0-1}, this time using~\eqref{eq:succ-4}. Note that the last sum in \Eq{s0-1} may be negative, and we need to bound it as well. We proceed as follows: 
\begin{align*}
  \Big|\sum_{(s,t)\in S_1}\tr(\rho_s R'_{s,t}\rho'_t) \Big|
  &\leq \Big(\sum_{(s,t)\in S'}\tr(\rho_s R'_{s,t}\rho'_t R'_{s,t}) \Big)^{1/2} \Big(\sum_{(s,t)\in S_1}\tr(\rho_s \rho'_t ) \Big)^{1/2}
\\
  &\leq \Big(\sum_{(s,t)\in S_1}\tr
  \big( \sqrt{R'_{s,t}}\rho_s \sqrt{R'_{s,t}}\big)\tr
  \big(\sqrt{R'_{s,t}}\rho'_t \sqrt{R'_{s,t}}\big) \Big)^{1/2} 
  \big(\tr(\rho^2)\big)^{1/2}
\\
  &\leq\Big(\sum_{(s,t)\in S_1} \tr( \rho_s U'_{s})
  \tr(\rho'_t P'_t) \Big)^{1/2} 
\\
  &\leq\Big(\Big(\sum_{s} \tr( \rho_s U'_{s})\Big)^{1/2}
  \Big(\sum_t \tr(\rho'_t P'_t) \Big)\Big)^{1/2}
\\
  &\leq \eps/\sqrt{\pi(A,B')\pi(A',B')},
\end{align*}
where the first inequality uses the Cauchy-Schwarz inequality for the trace inner product (using cyclicity of the trace to write $\tr(\rho_s R'_{s,t}\rho'_t) = \tr(\rho_s^{1/2} R'_{s,t}(\rho'_t)^{1/2}\cdot (\rho'_t)^{1/2}\rho_s^{1/2})$), the second uses the inequality $\tr(XY)\leq\tr(X)\tr(Y)$ valid for all positive $X,Y$, and the fact that $\tr(\rho_s\rho'_t)\geq 0$ for every $(s,t)$, the third uses cyclicity of the trace and $R'_{s,t}\leq U'_s$, $R'_{s,t}\leq P'_t$, the fourth uses that all terms are non-negative to upper bound the summation over $(s,t)\in S'$ by a summation over all $s,t$, and the last uses~\eqref{eq:succ-2} and~\eqref{eq:succ-4}. 

Putting everything together, we have shown that $\big|\sum_{(s,t)\in S'} \tr(\rho_s\rho_t)\big| = O(\sqrt{\eps})$, where the implied constant depends only on $\min_{XY}\pi(X,Y)$. Proceeding in a similar way, the same bound can be obtained for the summations over pairs $(s,t)\in S$. From \Eq{def-rho} we then obtain
$$ \frac{1}{d} \,\leq\, \sum_{s,t} \tr(\rho_s\rho'_t) \,=\, O\big(\sqrt{\eps}\big),$$
proving the bound claimed in the theorem. 
\end{proof}

\subsection{Strategies with bounded number of answers}

\begin{theorem}
Let $G$ be a two-player game with two questions per player and finite answer sets $\outA$ and $\outB$. 
Also let $\eps>0$ and $0<\delta<\omega^*(G)-\omega(G)$.
Consider an entangled strategy that succeeds at $\Ginf$ with probability at least $1-\eps$ and for which either Alice's or Bob's answers always fall within a set of cardinality~$L$.
Then~$L=\Omega(1/\sqrt{\eps})$, where the implied constant depends on $\min_{XY}\pi(X,Y)$ only. 
\label{thm:Answers}
\end{theorem}

%\begin{theorem} Let $\eps>0$, and let $L$ be an integer. Consider a quantum strategy that wins Hardy's game with probability at least $1-\eps$ and for which either Alice's or Bob's answers always fall within a set of cardinality $L$.
%Then~$L \geq 1/(16\sqrt{\eps})$.
%\label{thm:Answers}
%\end{theorem}

\begin{proof}
Let $L\in\N$ and fix a string $s\in\outA^*$ of length at most $\log(L)$ that is answered by Alice on question $A$. Let $\sigma_0$ be the post-measurement state on Alice's system that corresponds to Bob applying his POVM $\set{B_u}$ upon question $B$ and obtaining an answer string $u\notin \ZZ_s$, that is, 
\begin{equation}
  \sigma_0:=\sum_{u\notin\ZZ_s} 
  \tr_A\big((\id\otimes B_u) \ketbras{\psi}\big).
\end{equation}
If the questions are $(A,B)$ and Alice obtains $s$ as outcome, then by definition of the set $\ZZ_s$ all other answers for Bob are rejected in the game, hence 
\begin{equation}\label{eq:rho-as}
\tr\big(A_s(\rho-\sigma_0)\big)\,=\, \tr(\rho_s)-\tr(A_s\sigma_0) \,\leq\,\eps/\pi(A,B),
\end{equation}
where $\set{A_s: s\in\outA^*}$ is the POVM measured by Alice upon the unprimed question $A$.
Similarly, let $\sigma_1'$ be the post-measurement state on Alice's system that corresponds to Bob applying his POVM for question $B'$ and obtaining an answer string $v\notin\ZO_s$. By the same reasoning,
\begin{equation}\label{eq:rho-bs}
\tr\big(A_s(\rho-\sigma'_1)\big)\,=\, \tr(\rho_s)-\tr(A_s\sigma_1')\,\leq\,\eps/\pi(A,B').
\end{equation}

Let $P'$ be the sum of those POVM elements for Alice's question $A'$ that correspond to answers $t$ such that $\ZO_s \cup \OO_t = \outB^*$, and
$Q' = \id-P'$. By \Lem{discriminate}, every POVM element in $Q'$ corresponds to an answer $t$ such that $\ZZ_s \cup \OZ_t = \outB^*$. On questions $(A',B)$, if Alice obtains one of the outcomes $t$ corresponding to $P'$ and Bob obtains an outcome $u\notin\ZO_s$ then it must be that $u\in \OZ_t$ and hence $\tr(P'\sigma_1') \leq \eps/\pi(A',B)$. Similarly, on questions $(A',B')$ if  if Alice obtains one of the outcomes $t$ corresponding to $Q'$ and Bob obtains an outcome $u\notin \ZZ_s$ then $u\in\OZ_t$ and hence $\tr(Q'\sigma_0 ) \leq \eps/\pi(A',B')$. We have proven:
\begin{equation}\label{eq:rho-cs}
\tr(Q'\sigma_0 ) \leq \eps/\pi(A',B')\qquad\text{and}\qquad\tr(P'\sigma_1') \leq \eps/\pi(A',B).
\end{equation}
Using $P'+Q'=\id$ we may decompose
\begin{align}
\tr(A_s \sigma_0) &= \tr( A_s  P' \sigma_0) + \tr(A_s  Q' \sigma_0)\notag\\
&\leq  \tr( A_s  P' \sigma_0) + \tr(A_s Q' A_s \sigma_0)^{1/2}\tr(  Q' \sigma_0)^{1/2}\notag\\
&\leq \tr( A_s  P' \sigma_0) +\sqrt{\eps/\pi(A',B')}\notag\\
&= \tr( A_s  P' \rho) - \tr( A_s  P' (\rho-\sigma_0))+\sqrt{\eps/\pi(A',B')}\notag\\
&\leq \tr( A_s  P' \rho)+  \tr(P' A_s P' (\rho-\sigma_0))^{1/2} \tr( A_s  (\rho-\sigma_0))^{1/2}+\sqrt{\eps/\pi(A',B')}\notag\\
&\leq \tr( A_s  P' \rho)+O\big(\sqrt{\eps}\big),\label{eq:rho-bs-1}
\end{align} 
where the first inequality is Cauchy-Schwarz, the second uses $A_s Q' A_s \leq \id$ and~\eqref{eq:rho-cs}, the third inequality is Cauchy-Schwarz again, and the last uses $P'A_sP'\leq \id$ and~\eqref{eq:rho-as}. A similar chain of inequalities leads to the bound 
\begin{equation}\label{eq:rho-bs-2}
\tr(A_s \sigma_1') \leq \tr( A_s  Q' \rho) + O\big(\sqrt{\eps}\big).
\end{equation}
But 
$$ \tr( A_s  P' \rho) + \tr(A_s  Q' \rho) = \tr( A_s \rho)=\tr(\rho_s),$$
which combined with~\eqref{eq:rho-bs-1},~\eqref{eq:rho-bs-2} and~\eqref{eq:rho-as},~\eqref{eq:rho-bs} yields 
$$2\tr(\rho_s)\,\leq\, \tr(A_s \sigma_0) + \tr(A_s\sigma_1') + O\big(\sqrt{\eps}\big)\, \leq\, \tr(\rho_s) + O\big(\sqrt{\eps}\big),$$
implying that $\tr(\rho_s) = O\big(\sqrt{\eps}\big)$. Summing over $s$ we deduce $1\leq O(\sqrt{\eps})L$, proving the theorem.   
\end{proof}

We can use Hardy's game, or any of the games $\Ginf$, to certify the dimension of an entangled state. According to \Thm{Finite}, if Alice and Bob succeed with high probability, then they must possess an entangled state of high local dimension. In practice it is not reasonable to test entanglement via a scheme involving infinitely many outcomes. Therefore, we consider games $\Ginf(\ell)$ that are obtained from $\Ginf$ by limiting the parties to answer strings of length at most $\ell\in \N$. For an appropriately chosen $\ell=\Theta(\log 1/\eps)$, the game $\Ginf(\ell)$ has two questions and $\text{poly}(1/\eps)$ answers per player; moreover \Lem{Strategy} shows that $\omega^*\big(\Ginf(\ell)\big)\ge 1-\eps$. Now, as a direct consequence of \Thm{Finite} we obtain the following corollary.

\begin{cor}
Let $G$ be a two-player game with two questions per player and finite answer sets $\outA$ and $\outB$. 
Also let $\eps>0$, $0<\delta<\omega^*(G)-\omega(G)$
and  $\ell=\Theta(\log 1/\eps)$. Any strategy that succeeds at $\Ginf(\ell)$ with probability at least $\omega^*\big(\Ginf(\ell)\big)-\eps$ must use a state of local dimension at least $\Omega(1/\sqrt{\eps})$.
\label{cor:Witness}
\end{cor}

\section{Discussion and Open Problems}
\label{sec:Conclusion}

By exhibiting Hardy's game we have shown that there exist two-party distributed tasks with classical questions and answers for which no finite amount of entanglement allows the parties to perform optimally. We have also exhibited an infinite sequence of strategies, using entanglement of increasing dimension, which obtain success probabilities that tend to the optimal $\omega^*(H)=1$. Each strategy in our sequence produces answers of fixed length, increasing with the dimension of the entangled state. We have showed that to some extent this is necessary:  strategies producing answers with bounded length succeed with probabilities that are bounded away from one. Finally, we have shown that games similar to Hardy's can be constructed out of any game $G$ with two questions per player, finite answer sets and $\omega(G)<\omega^*(G)$.

The most pressing question left open by our work is whether a finite-dimensional shared state is always sufficient to achieve $\omega^*(G)$ for games with finite question and answer sets. As already mentioned in the introduction, there are numerical results which suggest that this is not always the case \cite{Pal10}. It is also interesting to investigate what  kind of trade-offs may be required in terms of the dimension of strategies which approach the optimum. In particular, are there any general upper bounds on the dimension $d$ sufficient to achieve value $\omega^*(G)-\eps$?

\paragraph{Acknowledgments.}
The first author thanks the anonymous referees of QIP 2014 for pointing out a mistake in an earlier proof of the non-quantitative part of \Thm{Finite}, Aleksandrs Belovs for suggesting the proof of \Lem{InfiniteInput} as well Richard Cleve, David Roberson, Stephanie Wehner, Harry Buhrman and his group at CWI for helpful discussions. Both authors are grateful to Oded Regev who prompted us to generalize our initial results, which were restricted to Hardy's game, by suggesting that the same construction might also apply for the CHSH game. Both authors are supported by the Ministry of Education, Singapore under the Tier 3 grant MOE2012-T3-1-009. Thomas Vidick is also supported by the Newton Institute, Cambridge. Parts of this work was completed during a workshop at the Institute for Mathematical Sciences, Singapore.

\bibliographystyle{alphaurl}
%\bibliography{MaxEnt}

\begin{thebibliography}{BCMdW10}

\bibitem[BBT11]{Briet11}
Jop Bri\"{e}t, Harry Buhrman, and Ben Toner.
\newblock A generalized {G}rothendieck inequality and nonlocal correlations
  that require high entanglement.
\newblock {\em Comm. Math. Phys.}, 305(3):827--843, 2011.
\newblock \href {http://arxiv.org/abs/0901.2009} {\path{arXiv:0901.2009}}.

\bibitem[BCMdW10]{Buhrman10}
Harry Buhrman, Richard Cleve, Serge Massar, and Ronald de~Wolf.
\newblock Nonlocality and communication complexity.
\newblock {\em Rev. Mod. Phys.}, 82:665--698, 2010.
\newblock \href {http://arxiv.org/abs/0907.3584} {\path{arXiv:0907.3584}}.

\bibitem[BCP{\etalchar{+}}14]{Brunner13Bell}
Nicolas Brunner, Daniel Cavalcanti, Stefano Pironio, Valerio Scarani, and
  Stephanie Wehner.
\newblock Bell nonlocality.
\newblock {\em Rev. Mod. Phys.}, 86:419--478, 2014.
\newblock \href {http://arxiv.org/abs/1303.2849} {\path{arXiv:1303.2849}}.

\bibitem[BPA{\etalchar{+}}08]{Brunner08}
Nicolas Brunner, Stefano Pironio, Antonio Acin, Nicolas Gisin, Andr\'e~Allan
  M\'ethot, and Valerio Scarani.
\newblock Testing the dimension of {H}ilbert spaces.
\newblock {\em Phys. Rev. Lett.}, 100:210503, May 2008.
\newblock \href {http://arxiv.org/abs/0802.0760} {\path{arXiv:0802.0760}}.

\bibitem[BW92]{BennettW92superdense}
Charles~H. Bennett and Stephen~J. Wiesner.
\newblock Communication via one- and two-particle operators on
  {E}instein-{P}odolsky-{R}osen states.
\newblock {\em Phys.~Rev.~Lett.}, 69:2881--2884, 1992.
\newblock \href {http://dx.doi.org/10.1103/PhysRevLett.69.2881}
  {\path{doi:10.1103/PhysRevLett.69.2881}}.

\bibitem[BYJK08]{Bar08}
Ziv Bar-Yossef, Thathachar~S. Jayram, and Iordanis Kerenidis.
\newblock Exponential separation of quantum and classical one-way communication
  complexity.
\newblock {\em SIAM J. Comput.}, 38(1):366--384, 2008.
\newblock \href {http://dx.doi.org/10.1137/060651835}
  {\path{doi:10.1137/060651835}}.

\bibitem[CHTW04]{Cleve04}
Richard Cleve, Peter Hoyer, Benjamin Toner, and John Watrous.
\newblock Consequences and limits of nonlocal strategies.
\newblock In {\em Proc. of CCC'04}, pages 236--249, 2004.
\newblock \href {http://arxiv.org/abs/quant-ph/0404076}
  {\path{arXiv:quant-ph/0404076}}.

\bibitem[Col06]{Colbeck09}
Roger Colbeck.
\newblock {\em Quantum And Relativistic Protocols For Secure Multi-Party
  Computation}.
\newblock PhD thesis, Trinity College, University of Cambridge, 2006.
\newblock \href {http://arxiv.org/abs/0911.3814} {\path{arXiv:0911.3814}}.

\bibitem[CY14]{CoudronY13infinite}
Matthew Coudron and Henry Yuen.
\newblock Infinite randomness expansion and amplification with a constant
  number of devices.
\newblock In {\em Proc. of STOC'14}, pages 427--436, 2014.
\newblock \href {http://arxiv.org/abs/1310.6755} {\path{arXiv:1310.6755}}.

\bibitem[Eke91]{Ekert91}
Artur~K. Ekert.
\newblock {Quantum cryptography based on Bell's theorem}.
\newblock {\em Phys.~Rev.~Lett.}, 67(6):661--663, 1991.
\newblock \href {http://dx.doi.org/10.1103/physrevlett.67.661}
  {\path{doi:10.1103/physrevlett.67.661}}.

\bibitem[EPR35]{epr}
Albert Einstein, Boris Podolsky, and Nathan Rosen.
\newblock Can quantum-mechanical description of physical reality be considered
  complete?
\newblock {\em Phys.~Rev.}, 47:777--780, 1935.

\bibitem[Fri12]{Fritz12}
Tobias Fritz.
\newblock Tsirelson's problem and {K}irchberg's conjecture.
\newblock {\em Rev. Math. Phys.}, 24(05):1250012, 2012.
\newblock \href {http://arxiv.org/abs/1008.1168} {\path{arXiv:1008.1168}}.

\bibitem[Gle57]{Gleason57}
Andrew~M. Gleason.
\newblock Measures on the closed subspaces of a {H}ilbert space.
\newblock {\em J. Math. Mech.}, 6:885--893, 1957.

\bibitem[Gol94]{Goldstein94}
Sheldon Goldstein.
\newblock Nonlocality without inequalities for almost all entangled states for
  two particles.
\newblock {\em Phys. Rev. Lett.}, 72:1951--1951, 1994.
\newblock \href {http://dx.doi.org/10.1103/PhysRevLett.72.1951}
  {\path{doi:10.1103/PhysRevLett.72.1951}}.

\bibitem[Har92]{Hardy92}
Lucien Hardy.
\newblock Quantum mechanics, local realistic theories, and {L}orentz-invariant
  realistic theories.
\newblock {\em Phys. Rev. Lett.}, 68:2981--2984, 1992.
\newblock \href {http://dx.doi.org/10.1103/PhysRevLett.68.2981}
  {\path{doi:10.1103/PhysRevLett.68.2981}}.

\bibitem[Har93]{Hardy93}
Lucien Hardy.
\newblock Nonlocality for two particles without inequalities for almost all
  entangled states.
\newblock {\em Phys. Rev. Lett.}, 71:1665--1668, 1993.
\newblock \href {http://dx.doi.org/10.1103/PhysRevLett.71.1665}
  {\path{doi:10.1103/PhysRevLett.71.1665}}.

\bibitem[IKM09]{ItoKM09}
Tsuyoshi Ito, Hirotada Kobayashi, and Keiji Matsumoto.
\newblock Oracularization and two-prover one-round interactive proofs against
  nonlocal strategies.
\newblock In {\em Proc. of CCC'09}, pages 217--228, 2009.
\newblock \href {http://arxiv.org/abs/0810.0693} {\path{arXiv:0810.0693}}.

\bibitem[KS67]{Kochen67}
Simon Kochen and Ernst~P. Specker.
\newblock The problem of hidden variables in quantum mechanics.
\newblock {\em J. Math. Mech.}, 17:59--87, 1967.

\bibitem[Lal01]{Laloe01}
Franck Laloe.
\newblock Do we really understand quantum mechanics? {S}trange correlations,
  paradoxes, and theorems.
\newblock {\em Am. J. Phys.}, 69:655, 2001.
\newblock \href {http://arxiv.org/abs/quant-ph/0209123}
  {\path{arXiv:quant-ph/0209123}}.

\bibitem[Lot02]{Lothaire}
M.~Lothaire.
\newblock {\em Algebraic Combinatorics on Words}.
\newblock Cambridge University Press, 2002.

\bibitem[LTW13]{Leung08}
Debbie Leung, Ben Toner, and John Watrous.
\newblock Coherent state exchange in multi-prover quantum interactive proof
  systems.
\newblock {\em Chic. J. Theor. Comput.}, 2013(11):1--18, 2013.
\newblock \href {http://arxiv.org/abs/0804.4118} {\path{arXiv:0804.4118}}.

\bibitem[NPA08]{Navascues08}
Miguel Navascu\'{e}s, Stefano Pironio, and Antonio Ac\'{i}n.
\newblock A convergent hierarchy of semidefinite programs characterizing the
  set of quantum correlations.
\newblock {\em New J. Phys.}, 10(7):073013, 2008.
\newblock \href {http://arxiv.org/abs/0803.4290} {\path{arXiv:0803.4290}}.

\bibitem[PAM{\etalchar{+}}10]{Pironio}
Stefano Pironio, Antonio Ac{\'\i}n, Serge Massar, A.~Boyer de~La~Giroday,
  Dzimitry~N. Matsukevich, Peter Maunz, Steven Olmschenk, David Hayes, Le~Luo,
  T.~Andrew Manning, and Cristopher Monroe.
\newblock Random numbers certified by {B}ell's theorem.
\newblock {\em Nature}, 464(7291):10, 2010.
\newblock \href {http://arxiv.org/abs/0911.3427} {\path{arXiv:0911.3427}}.

\bibitem[PV10]{Pal10}
K\'aroly~F. P\'al and Tam\'as V\'ertesi.
\newblock Maximal violation of a bipartite three-setting, two-outcome {B}ell
  inequality using infinite-dimensional quantum systems.
\newblock {\em Phys. Rev. A}, 82:022116, 2010.
\newblock \href {http://arxiv.org/abs/1006.3032} {\path{arXiv:1006.3032}}.

\bibitem[RV13]{Regev12}
Oded Regev and Thomas Vidick.
\newblock Quantum xor games.
\newblock In {\em Proc. of CCC'13}, pages 144--155, 2013.
\newblock \href {http://arxiv.org/abs/1207.4939} {\path{arXiv:1207.4939}}.

\bibitem[Slo11]{Slofstra11}
William Slofstra.
\newblock Lower bounds on the entanglement needed to play {XOR} non-local
  games.
\newblock {\em J. Math. Phys.}, 52(10), 2011.
\newblock \href {http://arxiv.org/abs/1007.2248} {\path{arXiv:1007.2248}}.

\bibitem[VV12]{VV12}
Umesh Vazirani and Thomas Vidick.
\newblock Certifiable quantum dice: or, true random number generation secure
  against quantum adversaries.
\newblock In {\em Proc. of STOC'12}, pages 61--76, 2012.
\newblock \href {http://arxiv.org/abs/1111.6054} {\path{arXiv:1111.6054}}.

\bibitem[WCD08]{Wehner08}
Stephanie Wehner, Matthias Christandl, and Andrew~C. Doherty.
\newblock Lower bound on the dimension of a quantum system given measured data.
\newblock {\em Phys. Rev. A}, 78:062112, 2008.
\newblock \href {http://arxiv.org/abs/0808.3960} {\path{arXiv:0808.3960}}.

\bibitem[Wil13]{Wilde11}
Mark~M. Wilde.
\newblock {\em Quantum Information Theory}.
\newblock Cambridge University Press, 2013.
\newblock \href {http://arxiv.org/abs/1106.1445} {\path{arXiv:1106.1445}}.

\end{thebibliography}
\newcommand{\etalchar}[1]{$^{#1}$}

\end{document}